\documentclass[a4paper,12pt]{article}
\usepackage{amssymb}
\usepackage{amsmath}
\usepackage{epsfig}

\usepackage{graphics}

\usepackage{latexsym}
\usepackage{rotating}
\title{Negative Hierarchy of Hydrodynamic Type Equations}

\author{ Kostyantyn Zheltukhin \thanks{
email: zheltukh@metu.edu.tr}\\
{\small Department of Mathematics},\\
{\small Middle East Technical University, 06800 Ankara, Turkey}}

\date{\nonumber}
\begin{document}
\maketitle
\date{\nonumber}

\baselineskip 17pt

%\setlength{\textwidth}{15.5cm} \setlength{\textheight}{220mm}

%\numberwithin{equation}{section}

\begin{abstract} The negative integrable hierarchies of shallow water waves and dispersionless Toda lattice equations are considered.  The integrability is shown by 
explicit construction of an infinite set of conservation laws.
\end{abstract}

\newtheorem{theorem}{Theorem}
\newtheorem{lemma}[theorem]{Lemma}
\newtheorem{corollary}[theorem]{Corollary}
\newtheorem{example}[theorem]{Example}
\newtheorem{remark}[theorem]{Remark}

\section{Introduction}

Let us consider  hydrodynamic type equations that admit a Lax representation on a Poisson algebra. In such a case one  uses the $R$-matrix formalism (see \cite {Li}, \cite{Blaz} and references there in) to construct  integrable hierarchies of hydrodynamic type equations. We will discuss  negative extension of such a hierarchies. In particularwe discuss negative hierarchies for 
shallow water waves and Todda lattice equations.

Such negative hierarchies were  introduced by Karasu et al. in \cite{Karasu} and Kupershmidt in \cite{Kup} for  KdV equation. For other integrable equations they were extensively studied 
 in \cite{1GurPek}-\cite{3GurPek}. They give interesting examples of $(2+1)$-dimensional integrable equations.
 
 The negative hierarchies are constructed by considering negative powers of recursion operator. Consider a $(1+1)$-dimensional evolution equation 
\begin{equation}\label{eq1}
U_{t}=\Delta(U,U_x,\dots) 
\end{equation}
with a recursion operator ${\mathcal R}$. Assume that $U$ also depends on a variable $y$ then we can construct  the following  hierarchy of equations 
\begin{equation}\label{neghier1}
U_{y}=a{\mathcal R}^n \sigma_0 + b{\mathcal R}^{-m} \sigma_1, \qquad n,m\in {\mathbb N},  \quad a,b\in {\mathbb R},
\end{equation} 
where $\sigma_0$, $\sigma_1$ are symmetries of the equation (\ref{eq1}). 
To avoid dealing with negative powers of the operator $\mathcal R$ we can rewrite the hierarchy  (\ref{neghier1}) as 
\begin{equation}
a{\mathcal R}^m U_{y} + b{\mathcal R}^n \sigma_0= \sigma_1, \qquad n,m\in {\mathbb N},\quad a,b\in {\mathbb R}.
\end{equation}
In particular, the case $m=1$ and $\sigma_1=U_t$ is of interest for us. That is we have an equation 
\begin{equation}\label{negativehier2}
U_t=a{\mathcal R} U_y +b{\mathcal R}^n \sigma_0, \qquad n\in {\mathbb N}, a,b\in {\mathbb R}.
\end{equation}
Generally, such  equations are non-local but can be transformed to local form by  a change of variables. 
In the case of considered hydrodynamic type equations we get an integrable equations of the form
\begin{equation}
A\tilde U_{xt} + B\tilde U_{xy}+ C\tilde U_{xx}=0, \qquad j=1,2,\dots N,
\end{equation}
where $\tilde U$ is a vector function and $A,B,C$ are matrix functions depending on $U,U_x,U_y,U_t$. 

The integrability of the  equation (\ref{negativehier2}) can be shown in different ways. In \cite{1GurPek,2GurPek,3GurPek} the authors  constructed multi-soliton solutions of the obtained equations.
In  our case we construct  an infinite set of conservation laws.

The paper  is organized as follows: in Section 2 we give preliminary definitions and in Section 3 we consider examples of negative hierarchies.

\section{ Hierarchy of Equations of Hydrodynamic Type}

Following \cite{Li}, let us consider the algebra of Laurent series 
\begin{equation}
{\mathcal A}=\left \{ \sum_{n=-\infty}^\infty  u_n(x)p^n\,:\, u_n(x) \quad \mbox{decreasing rapidly as} \quad x\to\pm \infty \right \}
\end{equation}
with the brackets
\begin{equation}
\{f,g\}_k=p^k\left(\frac{\partial f}{\partial p}\frac{\partial g}{\partial x}-\frac{\partial f}{\partial x}\frac{\partial g}{\partial p}\right),\qquad k\in{\mathbb Z}.
\end{equation}
The algebra ${\mathcal A}$ admits a decomposition  into direct sum of sub-algebras  
\begin{equation}
{\mathcal A}={\mathcal A}_{\geq  -k+1}\oplus {\mathcal A}_{\leq -k},
\end{equation}
where $\displaystyle{{\mathcal A}_{\geq -k+1}=\left \{ \sum_{n=-k+1}^\infty  u_n(x)p^n  \right \}}$ and 
$\displaystyle{{\mathcal A}_{\leq -k}= \left \{ \sum_{n=-\infty}^{-k} u_n(x)p^n \right \}}$.
Taking the projection operators 
\begin{equation}
 {\mathcal P}_{\geq -k+1}: {\mathcal A}\to {\mathcal A}_{\geq -k+1} \quad \mbox{and} \quad   {\mathcal P}_{\leq -k}: {\mathcal A}\to {\mathcal A}_{\leq -k}
\end{equation}
we  construct an $R$-matrix
\begin{equation}
 {\mathcal K}=\frac{1}{2}\left({\mathcal P}_{\geq -k+1}-{\mathcal P}_{\leq -k}\right).
\end{equation}
Using the $R$-matrix, we write a Lax equation
\begin{equation}\label{R-mateq}
L_{t}=\{{\mathcal K} (L^q),L\}_k, \qquad L^q,L\in{\mathcal A}.
\end{equation}
It is convenient to denote  ${\mathcal P}_{\geq -k+1}(L)=(L)_{\geq -k+1}$, so the equation  (\ref{R-mateq}) becomes
\begin{equation}
\displaystyle
\frac{\partial L}{\partial t_n}=\{(L^q)_{\geq -k+1},L\}_k.
\end{equation}
To have a finite number of fields we take a Lax function
\begin{equation}\label{laxfunc}
L=p^{N-1}+\sum_{i=-1}^{N-2}m_i(x,t)p^i \qquad N\in{\mathbb N}.
\end{equation}
Now we can write a hierarchy of Lax equations 
\begin{equation}\label{laxeqnhydro}
\displaystyle
\frac{\partial L}{\partial t_n}=\{(L^{\frac{j}{N-1} +n})_{\geq -k+1},L\}_k, \qquad n=1,2,\dots
\end{equation}
for each  $j=1,2,\ldots, (N-1)$.
The above hierarchy admits the multi-Hamiltonian representation, a recursion operator and an infinite set of conservation laws, see \cite{GurZhel1} and references there in for more details. 
In particular, we have an infinite set of conservation laws with the conserved quantities
\begin{equation}\label{conserved}
T_n=\int (L^{\frac{j}{N-1} +n})_{=k-1}  dx, \quad n\in \mathbb N.
\end{equation}
We will use this conserved quantities in our construction. 

\noindent
To study some examples of negative hierarchies we take $N=2$, for simplicity, so
\begin{equation}\label{Lax2.0}
L=p+p+qp^{-1}
\end{equation}
and consider cases $k=0$ and $k=1$. In the case $k=0$, the equation (\ref{laxeqnhydro}) with Lax function (\ref{Lax2.0}) gives the shallow water waves equation (see \cite{Be})
\begin{align}\displaystyle
&\frac{1}{2}p_t=pp_x+q_x,\label{shallow1}\\
&\frac{1}{2}q_t=pq_x+qp_x\label{shallow2}
\end{align}
and 
in the case $k=1$, the equation (\ref{laxeqnhydro}) with Lax function (\ref{Lax2.0}) gives the dispersionless Toda lattice equations  (see \cite{St})
\begin{align}\displaystyle
&\frac{1}{2}p_t=q_x,\label{todda1}\\
&\frac{1}{2}q_t=pq_x+q\label{todda2}.
\end{align}
The other cases can be considered in a similar way. 

\section{Negative Hierarchy of Equations of Hydrodynamic Type}

To derive  negative hierarchies corresponding to the equation (\ref{laxeqnhydro}) it is convenient to use so called symmetric variables. We set the Lax function (\ref{Lax2.0}) as 
\begin{equation}\label{Lax2}
L=\frac{1}{p}(p-u)(p-v),
\end{equation}
where $u,v$ are new variables.

\subsection{Shallow Water Wave Equation ($k=0$ case)}
                                                                                                                                     
 Let $k=0$. In symmetric variables the system of equations (\ref{shallow1}), (\ref{shallow2}) is
\begin{align}\displaystyle
&\frac{1}{2}u_t=(u+v)u_x+uv_x,\label{symk=0N=2a}\\
&\frac{1}{2}v_t=vu_x+(u+v)v_x\label{symk=0N=2b}.
\end{align}
It admits  the recursion operator (see \cite{GurZhel1})
\begin{equation}\label{rec2}\displaystyle
 \mathcal{R}=\left( \begin{array}{cc}
u+v+u_xD_x^{-1} & 2u+u_xD_x^{-1} \\
2v+v_xD_x^{-1} & u+v+v_xD_x^{-1}
 \end{array} \right).
 \end{equation} 
and a set of conserved quantities, given by (\ref{conserved}),
\begin{equation}
\int T_n(u,v) dx, \qquad n\in {\mathbb N},
\end{equation}
where
\begin{equation}\label{conserved_1}
 T_n(u,v)=(L)^n_{=-1}=\sum_{k=1}^{n} C_n^kC_n^{k-1}u^kv^{n-k+1}.
\end{equation}
The functions $T_n$ have the following property.

\begin{lemma} The function $T_n$, given by (\ref{conserved_1}), satisfies
\begin{equation}\label{Tn_property_1}
v\frac{\partial^2}{\partial v^2} T_n=u\frac{\partial^2}{\partial u^2}T_n, \qquad \forall n\in {\mathbb N}.
\end{equation}
\end{lemma}
The lemma is proved by direct calculations.

\noindent
Let us study the equation (\ref{negativehier2}). First we take $a=1$ and $b=0$, then we get
\begin{equation}
U_t=\mathcal{R}(U_y),\quad \mbox{where} \quad
U=\left( \begin{array}{c}
u \\
v
 \end{array} \right).
\end{equation} 
 The corresponding equations are
 \begin{align}
 &u_t=uu_y+vu_y+2uv_y+u_xD_x^{-1}u_y+u_xD_x^{-1}v_y,\\
 &v_t=vv_y+uv_y+2vu_y+v_xD_x^{-1}v_y+v_xD_x^{-1}u_y.
 \end{align}
 To get rid of the term $D_x^{-1}$, we let $u=P_x, v=Q_x$. Then we have
\begin{align}
 P_{xt}=&P_xP_{xy}+Q_xP_{xy}+2P_xQ_{xy}+P_{xx}P_y+P_{xx}Q_y,\label{negsymn=1k=0N=2a}\\
 Q_{xt}=&Q_xQ_{xy}+P_xQ_{xy}+2Q_xP_{xy}+Q_{xx}Q_y+Q_{xx}P_y.\label{negsymn=1k=0N=2b}
 \end{align}
The integrability of the above equations follows from the theorem below.

\begin{theorem}
The system of equations (\ref{negsymn=1k=0N=2a}), (\ref{negsymn=1k=0N=2b}) admits an infinite set of conservation laws.
\end{theorem}
\noindent
{\bf Proof.} For each $n\in \mathbb N$ we take function $T_n(P_x,Q_x)$, where  $T_n$ is given by (\ref{conserved_1}). We multiply equation (\ref{negsymn=1k=0N=2a}) by $\frac{\partial T_n}{\partial P_x}$
and  equation (\ref{negsymn=1k=0N=2b}) by $\frac{\partial T_n}{\partial Q_x} $ and take the sum. By rearranging terms of the sum we get
 \begin{align}
 & P_{xt}\frac{\partial T_n}{\partial P_x}+Q_{xt}\frac{\partial T_n}{\partial Q_x}=(P_y+Q_y)P_{xx}\frac{\partial T_n}{\partial P_x}+(Q_y+P_y)Q_{xx}\frac{\partial T_n}{\partial Q_x}  \nonumber \\
 & + \left((P_x+Q_x)\frac{\partial T_n}{\partial P_x}+2Q_x\frac{\partial T_n}{\partial Q_x}\right)P_{xy} +
  \left((Q_x+P_x)\frac{\partial T_n}{\partial Q_x}+2P_x\frac{\partial T_n}{\partial P_x}\right)Q_{xy}.
 \end{align}
 Thus we have
 \begin{align}
 & \frac{\partial}{\partial t} T_n= \frac{\partial}{\partial x}( T_n(P_y+Q_y))  \nonumber \\
 & + \left((P_x+Q_x)\frac{\partial T_n}{\partial P_x}+2Q_x\frac{\partial T_n}{\partial Q_x}-T_n\right)P_{xy} +
  \left((Q_x+P_x)\frac{\partial T_n}{\partial Q_x}+2P_x\frac{\partial T_n}{\partial P_x}-T_n\right)Q_{xy}.
 \end{align} 
 We claim that there exists a function $S_n(P_x,Q_x)$ such that  
 \begin{equation}\label{partial_S_1}
 \frac{\partial S_n}{\partial P_x}= \left((P_x+Q_x)\frac{\partial T_n}{\partial P_x}+2Q_x\frac{\partial T_n}{\partial Q_x}-T_n\right) \quad \mbox {and} \quad
 \frac{\partial S_n}{\partial Q_x}= \left((Q_x+P_x)\frac{\partial T_n}{\partial Q_x}+2P_x\frac{\partial T_n}{\partial P_x}-T_n\right).
 \end{equation} 
To show the existence of $S_n$, it is enough to check that
\begin{equation}\label{mixed_equal_1}
 \frac{\partial }{\partial Q_x} \left((P_x+Q_x)\frac{\partial T_n}{\partial P_x}+2Q_x\frac{\partial T_n}{\partial Q_x}-T_n\right)=
 \frac{\partial }{\partial P_x}\left((Q_x+P_x)\frac{\partial T_n}{\partial Q_x}+2P_x\frac{\partial T_n}{\partial P_x}-T_n\right).
 \end{equation}
 We have 
\begin{equation}\nonumber
  \frac{\partial }{\partial P_x} \left((P_x+Q_x)\frac{\partial T_n}{\partial P_x}+2Q_x\frac{\partial T_n}{\partial Q_x}-T_n\right)= 
 \frac{\partial T_n}{\partial P_x} + \frac{\partial T_n}{\partial Q_x} +(P_x+Q_x)\frac{\partial^2 T_n}{\partial P_x\partial Q_x}+2Q_x\frac{\partial^2 T_n}{\partial Q_x^2}
 \end{equation}
 and
 \begin{equation}\nonumber
 \frac{\partial }{\partial P_x}\left((Q_x+P_x)\frac{\partial T_n}{\partial Q_x}+2P_x\frac{\partial T_n}{\partial P_x}-T_n\right)= 
 \frac{\partial T_n}{\partial Q_x} + \frac{\partial T_n}{\partial P_x} +(P_x+Q_x)\frac{\partial^2 T_n}{\partial Q_x\partial P_x}+2P_x\frac{\partial^2 T_n}{\partial P_x^2}.
 \end{equation}
 Since 
\begin{equation}\nonumber
Q_x\frac{\partial^2 T_n}{\partial Q_x^2}= P_x\frac{\partial^2 T_n}{\partial P_x^2},
\end{equation}
by equality (\ref{Tn_property_1}), the equality (\ref{mixed_equal_1}) holds. Thus there exists a function $S_n$ satisfying (\ref{partial_S_1}) and for each $n \in {\mathbb N}$ we can write a conservation law
\begin{equation}\label{law_1}
\frac{\partial}{\partial t} T_n= \frac{\partial}{\partial x}( T_n(P_y+Q_y)) + \frac{\partial}{\partial y}( S_n). 
\end{equation} 
$\Box$

\noindent
Let us calculate several conservation laws for equations (\ref{negsymn=1k=0N=2a}), (\ref{negsymn=1k=0N=2b}).
 
\begin{example}
For $n=1$, by (\ref{conserved_1}), we have $T_1(P_x,Q_x)=P_xQ_x$. Now the equality (\ref{partial_S_1}) implies that
\begin{equation}
 \frac{\partial S_1}{\partial P_x}= (P_x+Q_x)Q_x+2Q_xP_x-Q_xP_x \quad \mbox {and} \quad
 \frac{\partial S_1}{\partial Q_x}=(Q_x+P_x)P_x+2P_xQ_x-P_xQ_x.
 \end{equation}
So, we find $S_1(P_x,Q_x)=P_xQ_x^2+P_x^2Q_x$.  Hence by formula (\ref{law_1}) we get the first conservation law of equations (\ref{negsymn=1k=0N=2a}), (\ref{negsymn=1k=0N=2b})
\begin{equation}\label{ex_law_11}
\frac{\partial}{\partial t} \Big(P_xQ_x\Big)= \frac{\partial}{\partial x}\Big( P_xQ_x(P_y+Q_y)\Big) + \frac{\partial}{\partial y}\Big( P_xQ_x(P_x+Q_x)\Big). 
\end{equation} 
 \end{example}
 
\begin{example}
For $n=2$, by  (\ref{conserved_1}), we have $T_2(P_x,Q_x)=2P_xQ_x^2+2P_x^2Q_x$. Now the equality (\ref{partial_S_1}) implies that
\begin{equation}
 \frac{\partial S_2}{\partial P_x}= (P_x+Q_x)(4P_xQ_x+2Q_x^2)+2Q_x(2P_x^2+4P_xQ_x)-(2P_xQ_x^2+2P_x^2Q_x)
\end{equation}
and
\begin{equation} 
 \frac{\partial S_2}{\partial Q_x}=(Q_x+P_x)(2P_x^2+4P_xQ_x)+2P_x(4P_xQ_x+2Q_x^2)-(2P_xQ_x^2+2P_x^2Q_x).
\end{equation}
So, we find $S_2(P_x,Q_x)=2P_x^3Q_x+2P_xQ_x^3+6P_x^2Q_x^2$.  Hence, by formula (\ref{law_1}), we get the second conservation law of equations (\ref{negsymn=1k=0N=2a}), (\ref{negsymn=1k=0N=2b})
\begin{align} 
\frac{\partial}{\partial t} \Big(2P_xQ_x^2+2P_x^2Q_x\Big)= & \frac{\partial}{\partial x}\Big((P_y+Q_y)(2P_xQ_x^2+2P_x^2Q_x)\Big) \nonumber\\
                           & + \frac{\partial}{\partial y}\Big( 2P_xQ_x(P_x^3+3P_x^2Q_x^2+Q_x^3)\Big). \label{ex_law_12}
 \end{align} 
 \end{example}
 
\begin{example}
For $n=3$, by (\ref{conserved_1}), we have $T_3(P_x,Q_x)=3P_x^3Q_x+9P_x^2Q_x^2+3P_xQ_x^3$. Now the equality (\ref{partial_S_1}) implies that
 \begin{align}
  \frac{\partial S_3}{\partial P_x}= & (P_x+Q_x)(9P_x^2Q_x+18P_xQ_x^2+3Q_x^3) +2Q_x(3P_x^3+18P_x^2Q_x+9Q_x^2P_x)\nonumber \\
                                   &- (3P_x^3Q_x+9P_x^2Q_x^2+3P_xQ_x^3)
 \end{align}
and
\begin{align}
  \frac{\partial S_3}{\partial Q_x}= & (P_x+Q_x)(3P_x^3+18P_x^2Q_x+9Q_x^2P_x) +2P_x(9P_x^2Q_x+18P_xQ_x^2+3Q_x^3)\nonumber \\
                                   &- (3P_x^3Q_x+9P_x^2Q_x^2+3P_xQ_x^3),
 \end{align}
and we find $S_3(P_x,Q_x)=3P_xQ_x^4+18P_x^2Q_x^3+18P_x^3Q_x^2+3P_x^4Q_x$.  Hence, by formula (\ref{law_1}), we get the third conservation law of equations (\ref{negsymn=1k=0N=2a}), (\ref{negsymn=1k=0N=2b})
\begin{align}
 &\frac{\partial}{\partial t} \Big(3P_x^3Q_x^4+9P_x^2Q_x^2+3P_xQ_x^3\Big)= \frac{\partial}{\partial x}\Big((P_y+Q_y)(3P_x^3Q_x^4+9P_x^2Q_x^2+3P_xQ_x^3)\Big) \nonumber\\
 &  + \frac{\partial}{\partial y}\Big( 3P_xQ_x(Q_x^3+6P_xQ_x^2+6P_x^2Q_x+P_x^3)\Big). \label{ex_law_13}
\end{align} 
 \end{example} 
 
 \noindent
{\bf Remark.} In the case  $b\ne 0$ and $n\in \mathbb N$ the resulting equation also has an infinite set of conservation laws that are easily constructed from conservation laws of
  equations (\ref{negsymn=1k=0N=2a}), (\ref{negsymn=1k=0N=2b}) and conservation laws of  equations (\ref {symk=0N=2a}), (\ref{symk=0N=2b}).

\subsection{Toda Lattice Equation ($k=1$ case)}
                                                                                                                                     
 Let $k=1$. In symmetric variables the system of equations (\ref{todda1}), (\ref{todda2})  is
\begin{align}
&u_t=uv_x,\label{symk=1N=2a}\\
&v_t=vu_x.\label{symk=1N=2b}
\end{align}
It admits  the recursion operator (see \cite{GurZhel1})
\begin{equation}\label{rec5}\displaystyle
 \mathcal{R}=\left( \begin{array}{cc}
u+v+uv_xD_x^{-1}u^{-1}&2u+uv_xD_x^{-1}v^{-1} \\
2v+vu_xD_x^{-1}u^{-1} & u+v+vu_xD_x^{-1}v^{-1}
 \end{array} \right),
 \end{equation}
 and a set of conserved quantities, given by (\ref{conserved}), 
\begin{equation}
\int T_n(u,v) dx, \qquad n\in{\mathbb N},
\end{equation}
where
\begin{equation}\label{conserved_2}
 T_n(u,v)=(L)^n_{=0}=\sum_{k=1}^{n} C_n^kC_n^{k}u^kv^{n-k}.
\end{equation}
The functions $T_n$ have the following property.

\begin{lemma} The function $T_n$, given by (\ref{conserved_2}), satisfies
\begin{equation}\label{Tn_property_2}
\frac{\partial}{\partial u}\left( u \frac{\partial}{\partial u}T_n \right)=\frac{\partial}{\partial v} \left( v \frac{\partial}{\partial v}T_n \right),\qquad \forall n\in {\mathbb N}.
\end{equation}
\end{lemma}
The lemma is proved by direct calculations. We will need a corollary of the above lemma.

\begin{corollary}\label{Gn2}
For any $n\in {\mathbb N}$, there exists a function $G_n(u,v)$ such that 
\begin{align}
 & \frac{\partial }{\partial u} G_n = v \frac{\partial}{\partial v}T_n, \nonumber\\
 & \frac{\partial }{\partial v} G_n =  u \frac{\partial}{\partial u}T_n, \qquad \forall n\in {\mathbb N}. \label{Gn_property_2}
\end{align}
\end{corollary}
Let us study the equation (\ref{negativehier2}). First we take $a=1$ and $b=0$, then we get
\begin{equation}
U_t=\mathcal{R}(U_y),\quad \mbox{where} \quad
U=\left( \begin{array}{c}
u \\
v
 \end{array} \right).
\end{equation} 
The corresponding equations are
\begin{align}
&u_t=uu_y+vu_y+2uv_y+uv_xD_x^{-1}(u^{-1}u_y+v^{-1}v_y),\\
&v_t=vv_y+uv_y+2vu_y+vu_xD_x^{-1}(u^{-1}u_y+v^{-1}v_y).
\end{align}
If we insert $u=e^{P_x}$ and $v=e^{Q_x}$ into the above system we obtain
\begin{align}
&P_{xt}=P_{xy}(e^{P_x}+e^{Q_x})+2Q_{xy}e^{Q_x}+Q_{xx}e^{Q_x}(P_y+Q_y),\label{negsymn=1k=1N=2a}\\
&Q_{xt}=Q_{xy}(e^{Q_x}+e^{P_x})+2P_{xy}e^{P_x}+P_{xx}e^{P_x}(P_y+Q_y).\label{negsymn=1k=1N=2b}
\end{align}
The integrability of the above equations follows from the theorem below.

\begin{theorem}
The system of equations (\ref{negsymn=1k=1N=2a}), (\ref{negsymn=1k=1N=2b}) admits an infinite set of conservation laws.
\end{theorem}

\noindent
{\bf Proof.} For each $n\in \mathbb N$ we take function $T_n(u,v)$, where  $T_n$ is given by (\ref{conserved_2}).
Let $T_{n1}(u,v)=\frac{\partial T_n}{\partial u}$, $T_{n2}(u,v)=\frac{\partial T_n}{\partial v}$ and $T_{n12}(u,v)=\frac{\partial^2 T_n}{\partial u\partial v}$.
 We multiply equation (\ref{negsymn=1k=1N=2a}) by $e^{P_x}T_{1n}(e^{P_x},e^{Q_x})$,  equation (\ref{negsymn=1k=1N=2b}) by $e^{Q_x}T_{2n}(e^{P_x},e^{Q_x})$ and take the sum. 
 By rearranging terms of the sum we get
 \begin{align}
 &  T_{n1}e^{P_x}P_{xt}+T_{n2}e^{Q_x}Q_{xt}= T_{n1}e^{P_x}e^{Q_x}Q_{xx}(P_y+Q_y)+T_{n2}e^{Q_x}e^{P_x}P_{xx}(P_y+Q_y) \nonumber\\
 &+ \Big( T_{n1}e^{P_x}(e^{P_x}+e^{Q_x})+2T_{n2}e^{Q_x}e^{P_x}\Big)P_{xy}+\Big( T_{n2}e^{Q_x}(e^{P_x}+e^{Q_x})+2T_{n1}e^{Q_x}e^{P_x}\Big)Q_{xy}\, .
 \end{align}
 Using function $G_n$, introduced in  Corollary \ref{Gn2} we have
 \begin{align}
 & \frac{\partial}{\partial t} T_n= \frac{\partial}{\partial x}\Big( G_n(P_y+Q_y)\Big)  + \Big( T_{n1}e^{P_x}(e^{P_x}+e^{Q_x})+2T_{n2}e^{P_x}e^{P_x}e^{Q_x}-G_n\Big)P_{xy}+\\ 
 &\Big( T_{2n}e^{Q_x}(e^{P_x}+e^{Q_x})+2T_{n1}e^{P_x}e^{Q_x}-G_n\Big)Q_{xy},
 \end{align} 
 where $G_n,T_n$  are functions of $e^{P_x},e^{Q_x}$.
 We claim that there exists a function $S_n(P_x,Q_x)$ such that  
 \begin{align}
 &\frac{\partial S_n}{\partial P_x}=  \Big( T_{n1}e^{P_x}(e^{P_x}+e^{Q_x})+2T_{n2}e^{Q_x}e^{P_x}-G_n \Big),\quad  \nonumber\\
 &\frac{\partial S_n}{\partial Q_x}= \Big( T_{n2}e^{Q_x}(e^{P_x}+e^{Q_x})+2T_{n1}e^{Q_x}e^{P_x}-G_n\Big)\label{partial_S_2} .
 \end{align} 
To show the existence of $S_n$, it is enough to check that
\begin{align}
 & \frac{\partial }{\partial Q_x}\Big( T_{n1}e^{P_x}(e^{P_x}+e^{Q_x})+2T_{n2}e^{Q_x}e^{P_x}-G_n \Big)=  \nonumber\\
 & \frac{\partial }{\partial P_x}\Big( T_{n2}e^{Q_x}(e^{P_x}+e^{Q_x})+2T_{n1}e^{Q_x}e^{P_x}-G_n\Big). \label{mixed_equal_2}
 \end{align}
 Using Corollary \ref{Gn2} we have 
 \begin{align}
 & \frac{\partial }{\partial Q_x}\Big( T_{n1}e^{P_x}(e^{P_x}+e^{Q_x})+2T_{2n}e^{Q_x}e^{P_x}-G_n \Big)=  \nonumber\\
 & T_{n12}e^{P_x}e^{Q_x}(e^{P_x}+e^{Q_x})+T_{n1}e^{P_x}e^{Q_x}+2T_{n12}e^{Q_x}e^{Q_x}e^{P_x}+2T_{n2}e^{Q_x}e^{P_x}-T_{n1}e^{Q_x}e^{P_x}
 \end{align}
 and
 \begin{align}
 & \frac{\partial }{\partial P_x}\Big( T_{n2}e^{Q_x}(e^{P_x}+e^{Q_x})+2T_{n1}e^{Q_x}e^{P_x}-G_n\Big)=  \nonumber\\
 & T_{n12}e^{P_x}e^{Q_x}(e^{P_x}+e^{Q_x})+T_{n2}e^{P_x}e^{Q_x}+2T_{n12}e^{Q_x}e^{P_x}e^{P_x}+2T_{n1}e^{Q_x}e^{P_x}-T_{n2}e^{Q_x}e^{P_x}.
 \end{align}
 Since 
\begin{equation}\nonumber
T_{n12}e^{Q_x}+T_{n2}=T_{n12}e^{P_x}+T_{n1},
\end{equation}
by  equality (\ref{Tn_property_2}), the equality (\ref{mixed_equal_2}) holds. Thus there exists a function $S_n$ satisfying (\ref{partial_S_2}) and for each $n \in {\mathbb N}$ we can write a conservation law
\begin{equation}\label{law_2}
\frac{\partial}{\partial t} T_n= \frac{\partial}{\partial x}( G_n(P_y+Q_y)) + \frac{\partial}{\partial y}( S_n), 
\end{equation} 
 where $G_n,T_n$ and $S_n$ are functions of $e^{P_x},e^{Q_x}$
$\Box$

\noindent
Let us calculate several conservation laws for equations (\ref{negsymn=1k=1N=2a}), (\ref{negsymn=1k=1N=2b}).
 
\begin{example}
For $n=1$, by (\ref{conserved_1}), we have $T_1(u,v)=u+v$. The equality (\ref{Gn_property_2}) implies that
\begin{equation}
 \frac{\partial G_1}{\partial u}= v \quad \mbox {and} \quad
 \frac{\partial G_1}{\partial v}=u.
 \end{equation}
So, we find $G_1=uv$. We set $u=e^{P_x}$ and $v=e^{Q_x}$. Now the equality (\ref{partial_S_2}) implies that
\begin{equation}
 \frac{\partial S_1}{\partial P_x}= e^{2P_x}+2e^{P_x}e^{Q_x}
\end{equation}
and
\begin{equation} 
 \frac{\partial S_1}{\partial Q_x}= e^{2Q_x}+2e^{P_x}e^{Q_x}.
\end{equation} 
So, we find $S_1=\frac{1}{2}e^{2P_x}+2e^{P_x+Q_x}+\frac{1}{2}e^{2Q_x}$.

\noindent
Hence, by formula (\ref{law_2}), we get the first conservation law of equations (\ref{negsymn=1k=1N=2a}), (\ref{negsymn=1k=1N=2b})
\begin{equation}\label{ex_law_21}
\frac{\partial}{\partial t} \Big(e^{P_x}+e^{Q_x}\Big)= \frac{\partial}{\partial x}\Big( e^{P_x+Q_x}(P_y+Q_y)\Big) + 
\frac{\partial}{\partial y}\Big(\frac{1}{2}e^{2P_x}+2e^{P_x+Q_x}+\frac{1}{2}e^{2Q_x} \Big). 
\end{equation} 
 \end{example}
 
\begin{example}
For $n=2$, by  (\ref{conserved_1},) we have $T_2(u,v)=u^2+4uv+v^2$.  
The equality (\ref{Gn_property_2}) implies that
\begin{equation}
 \frac{\partial G_2}{\partial u}= 4uv+2v^2 \quad \mbox {and} \quad
 \frac{\partial G_2}{\partial v}=2u^2+4uv.
 \end{equation}
So, we find $G_2=2u^2v+2uv^2$. We set $u=e^{P_x}$ and $v=e^{Q_x}$. Now the equality (\ref{partial_S_2}) implies that
\begin{equation}
 \frac{\partial S_2}{\partial P_x}= 2e^{3P_x}+12e^{2P_x+Q_x}+6e^{P_x+2Q_x}
\end{equation}
and
\begin{equation} 
 \frac{\partial S_2}{\partial Q_x}= 2e^{3Q_x}+12e^{2Q_x+P_x}+6e^{Q_x+2P_x}.
\end{equation} 
So, we find $S_2=\frac{2}{3}e^{3P_x}+6e^{2P_x+Q_x}+6e^{P_x+2Q_x}+\frac{2}{3}e^{3Q_x}$.

\noindent
Hence, by formula (\ref{law_2}), we get the second conservation law of equations (\ref{negsymn=1k=1N=2a}), (\ref{negsymn=1k=1N=2b})
\begin{align}
& \frac{\partial}{\partial t} \Big(e^{2P_x}+4e^{Q_x+P_x}+e^{2Q_x}\Big)= \frac{\partial}{\partial x}\Big( (2e^{2P_x+Q_x}+2e^{P_x+2Q_x})(P_y+Q_y)\Big) \nonumber\\
& + \frac{\partial}{\partial y}\Big(\frac{2}{3}e^{3P_x}+6e^{2P_x+Q_x}+6e^{P_x+6Q_x}+\frac{2}{3}e^{3Q_x}\Big). \label{ex_law_22}
\end{align} 
 \end{example}
 
\begin{example}
For $n=3$, by  (\ref{conserved_1}), we have $T_3(u,v)=u^3+9u^2v+9uv^2+v^3$.  
The equality (\ref{Gn_property_2}) implies that
\begin{equation}
 \frac{\partial G_3}{\partial u}= 9u^2v +18uv^2+3v^3\quad \mbox {and} \quad
 \frac{\partial G_3}{\partial v}=3u^3+18u^2v+9uv^2.
 \end{equation}
So, we find $G_3=3u^3v+9u^2v^2+3uv^3$. We set $u=e^{P_x}$ and $v=e^{Q_x}$.  Now the equality (\ref{partial_S_2}) implies that
\begin{equation}
 \frac{\partial S_3}{\partial P_x}= 3e^{4P_x}+36e^{3P_x+Q_x}+54e^{2P_x+2Q_x}+12e^{P_x+3Q_x},
\end{equation}
and
\begin{equation} 
 \frac{\partial S_3}{\partial Q_x}= 12e^{3P_x+Q_x}+54e^{2P_x+2Q_x}+36e^{P_x+3Q_x}+3e^{4Q_x}
\end{equation} 
So, we find $S_3=\frac{3}{4}e^{4P_x}+12e^{3P_x+Q_x}+27e^{2P_x+2Q_x}+12e^{P_x+3Q_x}+\frac{3}{4}e^{4Q_x}$.

\noindent
Hence, by formula (\ref{law_2}), we get the third conservation law of equations (\ref{negsymn=1k=1N=2a}), (\ref{negsymn=1k=1N=2b})
\begin{align}
& \frac{\partial}{\partial t} \Big(e^{3P_x}+9e^{2P_x+Q_x}+9e^{P_x+2Q_x}+e^{3Q_x}\Big)= \nonumber\\
& \frac{\partial}{\partial x}\Big( (3e^{3P_x+Q_x}+9e^{2P_x+2Q_x}+3e^{P_x+2Q_x})(P_y+Q_y)\Big) \nonumber\\
& + \frac{\partial}{\partial y}\Big(\frac{3}{4}e^{4P_x}+12e^{3P_x+Q_x}+27e^{2P_x+2Q_x}+12e^{P_x+3Q_x}+\frac{3}{4}e^{4Q_x}\Big). \label{ex_law_23}
\end{align}
 \end{example} 
 
{\bf Remark.} In the case  $b\ne 0$ and $n\in \mathbb N$ the resulting equation has an infinite set of conservation laws that are constructed from conservation laws of
  equations (\ref{negsymn=1k=1N=2a}), (\ref{negsymn=1k=1N=2b}) and conservation laws of  equations (\ref {symk=1N=2a}), (\ref{symk=1N=2b}).

\end{document}